%% file: macroChaos.tex
\documentclass{article}

\usepackage{amssymb, amsmath, amsthm, tikz, verbatim, graphicx,lmodern,indentfirst,enumerate,color}

\usepackage[titletoc,page]{appendix}
\usepackage[margin=1.4in]{geometry}

\usepackage{cite}
\usepackage[skip=0pt]{caption}

\newtheorem{thm}{Theorem}[section]
\newtheorem{lem}[thm]{Lemma}
\newtheorem{prop}[thm]{Proposition}

\theoremstyle{definition}

\theoremstyle{definition}
\newtheorem{defn}[thm]{Definition}

\theoremstyle{remark}
\newtheorem{rem}[thm]{Remark}

\numberwithin{equation}{section}

\makeatletter

\newcommand{\Rmnum}[1]{\expandafter\@slowromancap\romannumeral #1@}
\newcommand\restr[2]{{
  \left.\kern-\nulldelimiterspace 
  #1 
  \right|_{#2} 
  }}

\makeatother
\begin{document}

\input{title}
\input{introduction}

\input{main}
\input{ergodic1}
\input{ergodic2}

\input{ergodicIterative}
\input{acknowledgements}
\bibliography{econbib}

\end{document}

%% file: title.tex
\title{Keynesian chaos revisited: odd period cycles and ergodic properties}
\author{Tomohiro Uchiyama\\
Faculty of International Liberal Arts, Soka University,\\ 
1-236 Tangi-machi, Hachioji-shi, Tokyo 192-8577, Japan\\
\texttt{Email:t.uchiyama2170@gmail.com}}
\date{}
\maketitle 



%% file: introduction.tex
\section{Introduction}\label{Introduction}

In this paper, we consider the (somewhat simplified) standard fixed-price macroeconomic dynamics of Mezler-Modigliani-Samuelson type generated by
\begin{equation}\label{dyn1}
Y_{t+1} = C(Y_t,r_t) + \mu I(Y_t,r_t). 
\end{equation}
Here, $Y_t\geq 0$ is the GDP level at time $t\in \mathbb{N}$, $r_t>0$ is the interest rate at time $t$, $C$ is a consumption function, $I$ is an investment function, and $\mu\geq 0$ is a constant that measures the strength of "induced investment" (in the sense of~\cite{DayShafer-Keynesian-Macro}). We assume no government or export/import. 

In a classical paper~\cite{DayShafer-Keynesian-Macro}, it was shown that if $\mu$ is large enough, Equation (\ref{dyn1}) exhibits a chaotic behaviour using the famous "Period three implies chaos" by Li-Yorke~\cite{LiYork-PeriodThree-Month}. This result of Li-Yorke gives a \emph{sufficient condition} for the existence of a \emph{Li-Yorke chaos} (see Definition~\ref{Li-Yorke}) and has been used a lot (possibly overused) in economic literature, see~\cite{BenhabibDay-Erratic-Letters},~\cite{Benhabib-Erratic-JEDC},~\cite{DayShafer-Keynesian-Macro}, and~\cite{NishimuraYano-Prog-soliton} for example. 

In the first part of this paper (Sections 1-4), we strengthen~\cite{DayShafer-Keynesian-Macro} by giving a \emph{necessary and sufficient} condition for the existence of a \emph{topological chaos} (or a \emph{turbulence}), see Definition~\ref{topChaos}, and see Theorems~\ref{main} and~\ref{main2} for precise statements. The point is that a topological chaos (or more precisely an odd period cycle) is more general than a (very special) period three cycle. This paper supports Mitra's attempt in~\cite{Mitra-topChaos-JET} to shift the current (too much) focus on Li-Yorke chaos (and too much reliance on Li-Yorke theorem) in economic literature to topological chaos (or topological entropies). See~\cite[Sec.~1]{Deng-TopChaos-JET} for more on this. Also see~\cite{Uchiyama-growth-ETB},~\cite{Uchiyama-OLG-OR}, and~\cite{Uchiyama-ergodic-arX} for characterisations for the existence of a topological chaos and their applications in different contexts.  

In the second part of this paper (Sections 5-7), we update~\cite{DayShafer-Ergodic-Macro}. In particular, using recent results on ergodic theory~\cite{Lyubich-forty-Modern},~\cite{Shen:2014} and some numerical argument, we show that even if a chaos exists, we can still predict the future GDPs "on average" (in the sense of ergodic theory, see Section 5 for a precise meaning of this). Here, we stress that a celebrated result by A. Avila (2014 fields medalist, see Proposition~\ref{Avila}) and others~\cite{Avila-unimodal-Annals},~\cite{Avila-stochastic-Invent} supports our theoretical/numerical argument for our "nonlinear model", see Section~6 for details. One of our purposes of this paper is to inject this deep mathematics into economic literature (which is not done yet elsewhere). 

Here is the structure of the rest of the paper. In Section~2, we set the stage for our argument. In particular, after giving all the necessary definitions, we explain our two models ("piecewise linear model" and "nonlinear model"). Then, in Section~3, we state our first main results (Theorems~\ref{main} and~\ref{main2}), that are the complete characterisations for the existence of (topological) chaos for our two models. In Section~4, we give proofs of our first main results. After that, in Sections~5 and~6, we investigate ergodic properties for our nonlinear model. In particular, in Section~5, after giving a short overview of (a tiny bit of) ergodic theory, we conduct some numerical calculations and show that we can predict the future GDP levels on average (Theorem~\ref{ergThm1-1}). In Section~6, we conduct a sensitivity analysis showing that how the future GDP levels vary when some parameter in the model changes (Theorem~\ref{finalThm}). In Section~7, we conduct a similar sensitivity analysis for our piecewise linear model (using a totally different mathematical technique) and obtain Theorem~\ref{iterativeThm}. 

All programming files used for numerical calculations and for generating plots in this paper (Jupyter notebooks) are available upon request.

\section{Preliminaries}\label{Preliminaries}
\subsection{Li-Yorke chaos vs topological chaos}\label{LiYorkeTopChaos}

Here, we clarify what we mean by "chaos" since there are several definitions of chaos in the literature, see~\cite{Ruette-book}. Along the way, we give all the necessary definitions used in this paper. The following definitions are taken from~\cite[Def.~5.1]{Ruette-book} and~\cite[Chap.~\Rmnum{2}]{Block-book}. Let $g$ be a continuous map of a closed interval $I$ into itself: 

\begin{defn}\label{Li-Yorke}
We say that $g$ exhibits a \emph{Li-Yorke chaos} if there exists an uncountable \emph{scrambled set} $S\subset I$, that is, for any $x,y\in S$ we have 
\begin{equation*}
\limsup_{n\rightarrow \infty}\mid g^n(x)-g^n(y)\mid > 0 \textup{ and } \liminf_{n\rightarrow \infty}\mid g^n(x)-g^n(y)\mid=0,
\end{equation*}
and for $x\in S$ and $y$ being a periodic point of $g$,
\begin{equation*}
\limsup_{n\rightarrow \infty}\mid g^n(x)-g^n(y)\mid > 0.
\end{equation*}
\end{defn}

\begin{defn}\label{topChaos}
We call $g$ \emph{turbulent} if there exist three points, $x_1$, $x_2$, and $x_3$ in $I$ such that $g(x_2)=g(x_1)=x_1$ and $g(x_3)=x_2$ with either $x_1<x_3<x_2$ or $x_2<x_3<x_1$. Moreover, we call $g$ \emph{(topologically) chaotic} if some iterate of $g$ is turbulent.  
\end{defn}
It is known that a map $g$ is topologically chaotic if and only if $g$ has a periodic point whose period is not a power of $2$, see~\cite[Chap.~\Rmnum{2}]{Block-book}. (So, in particular, if $g$ has an odd period cycle, then $g$ is topologically chaotic.) Note that this implies that a map $g$ is topologically chaotic if and only if the topological entropy of $g$ is positive, see~\cite[Chap.~\Rmnum{8}]{Block-book}. See~\cite{Ruette-book} for more characterisations and (subtle) mutual relations of various kinds of chaos (and of entropies). It is well-known that if $g$ has an odd period cycle, $g$ exhibits a Li-Yorke chaos, see~\cite[Fig.~0.1]{Ruette-book}.

\subsection{On the choice of our models}\label{OnChoiceModels}
Two models in this paper (as given below) are standard but very specific. Here are two main reasons why we have picked such specific models (rather than more generic models as in~\cite{DayShafer-Keynesian-Macro}). 

First, as we show in Sections~3 and~4, that obtaining a necessary and sufficient condition (for the existence of a topological chaos) is much more involved than obtaining a mere sufficient condition (for the existence of a Li-Yorke chaos as done in~\cite{DayShafer-Keynesian-Macro}). So if we use generic models, the result would look complicated and somewhat vague. Therefore, we find it better to use standard but specific models (with sharp results) to exhibit our method. In this paper, we have used two simplest possible standard models to obtain interesting (chaotic and ergodic) results. A similar analysis can be done with more complicated (and possibly more realistic) models. 

Second, our purpose of this paper is for applications (not for the theory itself): for example, by the concrete (algebraic) characterisations (for the existence of topological chaos) in Theorems~\ref{main} and~\ref{main2}, we were able to conduct sensitivity analyses, that are, how each parameter in the models affects our results. Also, for the second part of the paper (ergodic part), we needed some very specific properties of the functions defining the dynamics (to apply some delicate ergodic theory). Using very generic models as in~\cite{DayShafer-Keynesian-Macro} (although they are beautiful on their own), these sorts of detailed analyses (for applications) were not possible. 

Finally, we point out the reason why we analyse two different models in this single paper. This is because we would like to show a sharp contrast between these two models in the second part of the paper, that is, each model requires a distinct technique from ergodic theory to establish its future predictability. See~Sections~5, 6 and~7 for details. 

\subsection{Our models}\label{OurModels}
Let $M_1=k Y$ be the transactions demand for money where $k$ is a positive constant. We write $M_2$ for the liquidity preference and we assume that $M_2=\frac{\lambda}{r}$ where $\lambda$ is a positive constant and $r>0$. Given a fixed money supply $M$, the money market clearance condition gives 
\begin{equation}\label{moneyMarket}
M = M_1 + M_2 = k Y + \lambda/r.
\end{equation} 
In this paper, we normalise $M=1$ to simplify the exposition. Now, from Equation (\ref{moneyMarket}), we obtain $r=\frac{\lambda}{1-k Y}$ (the LM curve). Since we assumed $r>0$ (and $Y\geq 0$ as usual), the last equation forces $0 \leq Y < \frac{1}{k}$ (and thus $r\geq\lambda$). 

Substituting $r=\frac{\lambda}{1-k Y}$ into Equation (\ref{dyn1}) and further assume $C(Y,r)=\beta Y$ for some $0<\beta<1$, we obtain the main equation we consider in this paper:
\begin{equation}\label{dyn2}
Y_{t+1} = \beta Y_t + \mu I\left(Y_t, \frac{\lambda}{1-k Y_t}\right) \textup{ where } 0\leq Y_t < 1/k.
\end{equation}
 
From Equation (\ref{dyn2}), we see that our model becomes piecewise linear or nonlinear if the corresponding investment function becomes piecewise linear or nonlinear respectively. First, in Section~\ref{OPCNonlinear}, we let $I(Y,r) = \frac{\alpha Y}{r}$ for some $0<\alpha<1$ (higher GDP levels induce more investments). Then, setting $\delta:= \frac{\mu\alpha}{\lambda}$ in Equation (\ref{dyn2}), we obtain the equation for our first model ("nonlinear model"):
\begin{equation}\label{dyn3}
Y_{t+1} =f(Y_t):= \beta Y_t  + \delta Y_t (1-k Y_t) \textup{ where } 0\leq Y_t < 1/k.
\end{equation}


Second, in Section~\ref{OPCPiecewise}, following~\cite{DayShafer-Keynesian-Macro}, we let 
\begin{equation}\label{DayShaferInvest}
I(Y, r) =
\begin{cases}
	\frac{r''-r}{r} & (0\leq r \leq r'')\\
         0 & (r'' < r).
\end{cases} 
\end{equation}
This investment function assumes two things: 1.~There exists an interest rate $r''$ such that if the interest rate is more than $r''$, then the investment level becomes zero, 2.~The GDP level has no direct effect on investment (as opposed to our "nonlinear model"). Setting $r=r''$ in $r=\frac{\lambda}{1-k Y}$ (the LM curve) and solving for $Y$, we obtain $Y=\frac{r''-\lambda}{r'' k}$. Now set $Y_0:=\frac{r''-\lambda}{r'' k}$. Substituting Equation (\ref{DayShaferInvest}) into Equation (\ref{dyn2}), after some simplification, we obtain the equation for our second model ("piecewise linear model"):
\begin{equation}\label{DayShaferDyn}
Y_{t+1} = h(Y_t) =
\begin{cases}
	\left(\beta - \frac{\mu k}{1-k Y_0}\right) Y_t + \frac{\mu}{1-k Y_0} - \mu & (0 \leq Y_t \leq Y_0)\\
	\beta Y_t & (Y_0 < Y_t < \frac{1}{k}).
\end{cases}
\end{equation} 
It is easy to check that $h$ is continuous. In the rest of the paper, we assume $0<Y_0$ (that is equivalent to $r''>\lambda$) to avoid a degenerate uninteresting (no chaos) case.  

\section{First main results}\label{FirstMainResults}
\subsection{Odd period cycles in the nonlinear model}\label{OPCNonlinear}

Before stating first main results (Theorems~\ref{main} and~\ref{main2}), we need a bit more preparation. First, we recall the following mathematical result characterising the existence of a topological chaos for a unimodal interval map~\cite[Thms.~2 and~3]{Deng-TopChaos-JET}, see Proposition~\ref{ChaosThm} below. Our Theorem~\ref{main} is a (direct but highly non-trivial as seen in Section~\ref{SecProofMain}) consequence (or a special case) of Proposition~\ref{ChaosThm}. Let $\mathfrak{G}$ be the set of continuous maps from a closed interval $[a, b]$ to itself so that an arbitrary element $g\in \mathfrak{G}$ satisfies the following two properties:
\begin{enumerate}
\item{there exists $m\in (a,b)$ with the map $g$ strictly increasing on $[a,m]$ and strictly decreasing on $[m,b]$ ($g$ has the unique maximum at $m$.)}
\item{$g(a)\geq a$, $g(b) < b$, and $g(x)>x$ for all $x\in (a,m]$.}
\end{enumerate}
For $g\in \mathfrak{G}$, let $\Pi:=\{x\in [m,b]\mid g(x)\in [m,b] \textup{ and } g^2(x)=x\}$. Now we are ready to state:
\begin{prop}\label{ChaosThm}
Let $g\in \mathfrak{G}$. The map $g$ has an odd-period cycle if and only if $g^2(m) < m$ and $g^3(m) < \max\{x\in \Pi\}$ and the second iterate $g^2$ is turbulent if and only if $g^2(m) < m$ and $g^3(m) \leq \min\{x\in \Pi\}$. 
\end{prop}

We keep the notation in Equation (\ref{dyn3}) such as $\beta$, $\delta$, and $k$. Let $E:=\left[0,\frac{1}{k}\right]$. Using Proposition~\ref{ChaosThm}, we obtain (the first part of) first main results. 
\begin{thm}\label{main}
Suppose that $\max\{2-\beta,\beta\} < \delta \leq 2-\beta+2\sqrt{1-\beta}$. Then the map $f$ in Equation (\ref{dyn3}) is a map from $E$ to itself and has the following properties:
\begin{enumerate}
\item{$f\in \mathfrak{G}$.}
\item{$f$ has an odd period cycle if and only if $3.68-\beta<\delta$ ($3.68$ is an approximation, see Section~\ref{SubProofMain} for details).}
\item{The second iterate $f^2$ is turbulent if and only if $3.68-\beta\leq\delta$.} 
\end{enumerate}
\end{thm}
Now, a few comments are in order: 1.~The upper bound for $\delta$ in Theorem~\ref{main}, that is $\delta \leq 2-\beta+2\sqrt{1-\beta}$, is there just to keep $0\leq Y<\frac{1}{k}$. So the upper bound is not really interesting. The real meat is in the lower bound. Roughly speaking, if $\delta$ is sufficiently large (that means $\mu\alpha$ is large or $\lambda$ is small), then a chaos exists. This is because $\delta$ controls the extent of nonlinearity in Equation (\ref{dyn3}). We find it surprising that our result is independent of $k$. Note that if one wants to find how each parameter value is sensitive to our result, it is easy: just take the partial derivative in each inequality with respect to each parameter. Remember that $\delta= \frac{\mu \alpha}{\lambda}$. So we can compute the sensitivity of our result with respect to $\mu, \alpha, \lambda$, and $\beta$ (not just with respect to $\beta$ and $\delta$).  

\subsection{Odd period cycles in the piecewise linear model}\label{OPCPiecewise}
This subsection is similar to the last subsection, so we are brief. Recall~\cite[Cor.~3]{Deng-TopChaos-JET} (this is a counterpart of Proposition~\ref{ChaosThm}, here we deal with a unimodal function with the unique minimum rather than maximum). Let $\tilde{\mathfrak{G}}$ be the set of continuous maps from $[a, b]$ to itself so that an arbitrary element $g\in \tilde{\mathfrak{G}}$ satisfies:
\begin{enumerate}
\item{there exists $m\in (a,b)$ with the map $g$ strictly decreasing on $[a,m]$ and strictly increasing on $[m,b]$.}
\item{$g(a)>a$, $g(b)\leq b$, and $g(x)<x$ for all $x\in[m,b)$.}
\end{enumerate}
For $g\in \tilde{\mathfrak{G}}$, let $\tilde{\Pi}:=\{x\in [a,m]\mid g(x)\in [a,m] \textup{ and } g^2(x)=x\}$. Now we state~\cite[Cor.~3]{Deng-TopChaos-JET}:
\begin{prop}\label{ChaosThm2}
Let $g\in \tilde{\mathfrak{G}}$. The map $g$ has an odd-period cycle if and only if $g^2(m) > m$ and $g^3(m) > \max\{x\in \tilde{\Pi}\}$ and the second iterate $g^2$ is turbulent if and only if $g^2(m) > m$ and $g^3(m) \geq \min\{x\in \tilde{\Pi}\}$. 
\end{prop}

We keep $E=\left[0, \frac{1}{k}\right]$ from the last subsection. Also, we keep the same notation from Equation~(\ref{DayShaferDyn}) such as $\beta$, $\mu$, $k$, and $Y_0$. Using Proposition~\ref{ChaosThm2}, we obtain
\begin{thm}\label{main2}
Suppose that $\frac{\beta(1-k Y_0)}{k}<\mu \leq \frac{1-k Y_0}{k^2 Y_0}$. Then the map $h$ in Equation (\ref{DayShaferDyn}) is a map from $E$ to itself and has the following properties:
\begin{enumerate}
\item{$h \in \tilde{\mathfrak{G}}$.}
\item{$h$ has an odd period cycle if and only if $\frac{(1+\beta)(1-k Y_0)}{k} < \mu$ and \\$\frac{(1-k Y_0)(1+2 \beta^2+\sqrt{4\beta^2+1})}{2\beta k} < \mu$.}
\item{The second iterate $h^2$ is turbulent if and only if $\frac{(1+\beta)(1-k Y_0)}{k} < \mu$ and \\$\frac{(1-k Y_0)(1+2 \beta^2+\sqrt{4\beta^2+1})}{2\beta k} \leq \mu$.} 
\end{enumerate}
\end{thm}
Note that in Theorem~\ref{main2}, it is not possible to simplify (combine) the conditions $\frac{(1+\beta)(1-k Y_0)}{k} < \mu$ and $\frac{(1-k Y_0)(1+2 \beta^2+\sqrt{4\beta^2+1})}{2\beta k} < \mu$ since which lower bound for $\mu$ is larger depends on parameter values (such as $\beta$, $k$, and $Y_0$). Here, the upper bound for $\mu$ is there just to keep $0\leq Y \leq 1/k$, so it is not really interesting. The important part is the lower bound. Roughly speaking, if $\mu$ is large enough (where $\mu$ controls the slope of $h$ for $0\leq Y \leq Y_0$), $h$ exhibits a chaos (since if the slope is large, $h$ gets small/large very quick). By the concrete algebraic form of the characterisation, we can conduct the sensitivity analysis with respect to each parameter in the model just by taking the partial derivatives. Finally, note that the size of the slope of $h$ becomes important to establish the ergodic property in the second part of this paper, see Section 7 for details.  

\begin{rem}\label{iterativeExample}
Since the conditions for $\mu$ in Theorem~\ref{main2} look complicated, one might suspect that there might not be any combination of the parameter values $\mu, \beta, k, Y_0$ satisfying all the conditions to generate a chaos. This is not the case: take $\beta=0.6$, $k=1.1$, $Y_0=0.2\; (<1/k\approx 0.91)$. Then
\begin{enumerate}
\item{$\frac{\beta(1-k Y_0)}{k}<\mu \leq \frac{1-k Y_0}{k^2 Y_0}$ becomes $0.42 < \mu \leq 3.22$.}
\item{$\frac{(1+\beta)(1-k Y_0)}{k} < \mu$ becomes $1.13 < \mu$.}
\item{$\frac{(1-k Y_0)(1+2 \beta^2+\sqrt{4\beta^2+1})}{2\beta k} < \mu$ becomes $1.93<\mu$.}
\end{enumerate}
So, one can take $\mu=3$ for example. We investigate a chaotic (and ergodic) behaviour of this particular example in Section 7. Note that it is true that for some combination of parameter values $\beta$, $k$, and $Y_0$, there is no $\mu$ value satisfying all the conditions to generate a chaos. This means that we must choose parameter values carefully to generate a chaos. This kind of detailed analysis was not done (and cannot be done with their general model) in~\cite{DayShafer-Keynesian-Macro} and~\cite{DayShafer-Ergodic-Macro}.
\end{rem}


%% file: main.tex
\section{Proofs of first main results}\label{SecProofMain}
\subsection{Proof of Theorem~\ref{main}}\label{SubProofMain}
In the following proof, most results follow from direct (but fairly complicated) algebraic calculations. We give some sketches of our manipulations while pointing some important steps out rather than writing all the details. All calculations can be checked by a computer algebra system, say, Magma~\cite{magma}, Python~\cite{10.5555/1593511}, etc.

First, we want $f\in \mathfrak{G}$. We have $f'(Y)=\beta + \delta(1-k Y)-\delta k Y$. Solving $f'(Y)=0$, we obtain $Y=\frac{\beta+\delta}{2\delta k}$. Write $s:=\frac{\beta+\delta}{2\delta k}$ to ease the notation. Since we want $0<s<1/k$, this forces $\beta<\delta$. It is clear that $f$ is continuous and strictly increasing on $[0,s]$ and strictly decreasing on $[s,1/k]$. Thus, $f$ is unimodal and has the unique maximum at $Y=s$, and its maximum value is $f(s)=\frac{(\beta+\delta)^2}{4\delta k}$. Since we want to make $f$ a map from $E=[0,1/k]$ to itself, we need $0\leq f(s) \leq 1/k$. Solving this inequality for $\delta$, we obtain $2-\beta-2\sqrt{1-\beta} \leq \delta \leq 2-\beta+2\sqrt{1-\beta}$. Note that this lower bound is not binding since we have $\beta<\delta$ already and $2-\beta-2\sqrt{1-\beta}\leq \beta$ is equivalent to $\beta(1-\beta)\geq 0$, that is trivially true. So far, we need
\begin{equation}\label{equ1}
\beta < \delta \leq 2-\beta+2\sqrt{1-\beta}.
\end{equation}

Now, to force $f\in \mathfrak{G}$, we further need: (1)~$f(0)\geq 0$, (2)~$f(1/k)<1/k$, and (3)~$f(Y)>Y$ for all $Y\in (0,s]$. We see that (1) is clear and that (2) is equivalent to $\beta<1$, that is obviously true. Finally, since $f$ is strictly concave and $f(0)=0$, (3) is equivalent to $f(s)>s$. By an easy caluculation, we have that $f(s)>s$ is equivalent to $\delta > 2-\beta$. Now, combining the last inequality and (\ref{equ1}) we obtain
\begin{lem}\label{firstLem}
If $\max\{2-\beta,\beta\} < \delta \leq 2-\beta+2\sqrt{1-\beta}$, then $f\in \mathfrak{G}$.
\end{lem}

In the following, we assume $\max\{2-\beta,\beta\} < \delta \leq 2-\beta+2\sqrt{1-\beta}$. Next, we show that 
\begin{lem}\label{secondLem}
$f^2(s) < s$ if and only if $1-\beta+\sqrt{5}<\delta$.
\end{lem}
\begin{proof}
This follows from a (slightly messy) direct computation. We have 
\begin{alignat}{2}
s - f^2(s) &= s - f(f(s)) \nonumber\\
               &= s - f\left(\frac{(\beta+\delta)^2}{4\delta k}\right)\nonumber\\
               &= \frac{(\delta+\beta)(\delta+\beta-2)\left(\delta^2+\beta^2-2\beta+\delta(2\beta-2)-4\right)}{16\delta k} \label{last}
\end{alignat}
We want (\ref{last}) to be positive. Using Python, we obtain four real roots of (\ref{last}), those are $-\beta$, $1-\beta-\sqrt{5}$, $2-\beta$, and $1-\beta+\sqrt{5}$. It is easy to see that the first two roots are negative and the last two are positive. Since (\ref{last}) is a fourth degree polynomial in $\delta$ and we have assumed $\delta>0$, the necessary and sufficient condition for (\ref{last}) to be positive is $0<\delta<2-\beta$ or $1-\beta+\sqrt{5}<\delta$. Note that the first inequality cannot hold since we have already assumed that $\max\{2-\beta,\beta\} < \delta \leq 2-\beta+2\sqrt{1-\beta}$. This finishes the proof.  
\end{proof}

Now, solving $f(Y)=Y$, we obtain the fixed points of $f$, that are $Y=0, \frac{\delta+\beta-1}{\delta k}$. Write $z:=\frac{\delta+\beta-1}{\delta k}$ to ease the notation. Note that the assumption $\max\{2-\beta,\beta\} < \delta \leq 2-\beta+2\sqrt{1-\beta}$ forces $\beta+\delta>2$, that implies $s<z<1/k$. Recall from Introduction that we have $\Pi=\{Y\in [s,1/k]\mid f(Y)\in [s,1/k] \textup{ and } f^2(Y)=Y\}$. Remember that we have assumed $1-\beta+\sqrt{5}<\delta$ (from Lemma~\ref{firstLem}). Now, we show that 
\begin{lem}\label{thirdLem}
$\Pi = \{z\}$ (the unique non-trivial fixed point of $f$).
\end{lem}
\begin{proof}
First, it is clear that $z\in \Pi$. Now, solving $f^2(Y)=Y$ (this calculation is slightly messy), we obtain
$Y=0, z, \frac{\delta+\beta+1 \pm\sqrt{\delta^2+2\beta\delta-2\delta+\beta^2-2\beta-3}}{2\delta k}$. Looking closely at the last expression, we have
\begin{equation*}
\frac{\delta+\beta+1 \pm\sqrt{\delta^2+2\beta\delta-2\delta+\beta^2-2\beta-3}}{2\delta k} = s + \frac{1 \pm\sqrt{\delta^2+2\beta\delta-2\delta+\beta^2-2\beta-3}}{2\delta k}.\end{equation*}
So, if we show that $1 -\sqrt{\delta^2+2\beta\delta-2\delta+\beta^2-2\beta-3}<0$, then we are done. A simple calculation shows that the last inequality is equivalent to $(\delta+\beta-1)^2>5$. The last inequality holds since our assumption $1-\beta+\sqrt{5}<\delta$ implies $\delta+\beta-1 > \sqrt{5}$. 
\end{proof}
Finally we show that
\begin{lem}\label{forthLem}
$f^3(s) < z$ if and only if $3.68 - \beta < \delta$. (Likewise $f^3(s) \leq z$ if and only if $3.68 - \beta \leq \delta$.)
\end{lem}
\begin{proof}
We just give a proof for the strict inequality case. (The weak inequality case is similar.) By a direct calculation (using Python), we see that $f^3(s)<z$ is equivalent to
\begin{equation*}
\frac{1}{256\delta k} (\delta+\beta-2)^4(\delta+\beta+2)(\delta^3+3\beta\delta^2-2\delta^2+3\beta^2\delta-4\beta\delta-4\delta+\beta^3-2\beta^2-4\beta-8) > 0.
\end{equation*}
It is clear that this inequality is equivalent to 
\begin{equation}\label{lastExpression}
\delta^3+3\beta\delta^2-2\delta^2+3\beta^2\delta-4\beta\delta-4\delta+\beta^3-2\beta^2-4\beta-8>0.
\end{equation}
Let $j(\delta);=\delta^3+3\beta\delta^2-2\delta^2+3\beta^2\delta-4\beta\delta-4\delta+\beta^3-2\beta^2-4\beta-8$. Then we have $j'(\delta)=3\delta^2+6\beta\delta-4\delta+3\beta^2-4\beta-4$. Solving $j'(\delta)=0$, we obtain $\delta=2-\beta, -\beta-2/3$. Since $j$ is a third degree polynomial in $\delta$ and the coefficient of $\delta^3$ is positive, $j$ takes a local maximum at $\delta=-\beta-2/3<0$ and a local minimum at $\delta=2-\beta>0$. Now we compute $g(-\beta-2/3)=-16<0$. So, we see that $g$ takes only one real root. Solving (\ref{lastExpression}) directly using Python, we obtain 
\begin{equation*}
\delta > -\beta+\frac{2}{3}-\frac{(-152+24\sqrt{33})^{1/3}}{3}-\frac{16}{3(-152+24\sqrt{33})^{1/3}}.
\end{equation*}  
Rounding the last expression up to two decimal places, we obtain $\delta>3.68-\beta$. 
\end{proof}

Note that $3.68>1+\sqrt{5}\approx{3.24}$. Now, thanks to Proposition~\ref{ChaosThm}, combining Lemmas~\ref{firstLem},~\ref{secondLem},~\ref{thirdLem}, and~\ref{forthLem}, we obtain Theorem~\ref{main}.

\subsection{Proof of Theorem~\ref{main2}}\label{SubProofMain2}
The following proof is similar to that of the last subsection, so we just give a skeleton. First, we want $h\in \tilde{\mathfrak{G}}$. It is easy to check that $h$ is continuous at $Y_0$ (the turning point), so $h$ is continuous on $[0,1/k)$. Since we want to make $h$ a function on a closed interval (to use Proposition~\ref{ChaosThm2}), we extend $h$ to $[0,1/k]$ (where $h(1/k)=\beta/k$). Technically speaking, this forces to introduce $r=+\infty$ via $r=\frac{\lambda}{1-k Y}$, but this does not really affect the following argument. 

Next, to make $h$ unimodal, we need $\beta-\frac{\mu k}{1-k Y_0}<0$ (the slope of $h$ on $[0,Y_0]$) since the slope of $h$ on $[Y_0, 1/k]$ is strictly positive (that is $\beta$). Thus we have 
\begin{equation}\label{Vshape}
\frac{\beta(1-k Y_0)}{k}< \mu.
\end{equation}  

Now, to make $h$ a function from $E:=[0,1/k]$ to itself, we further need: (1)~$h(0)\leq 1/k$, (2)~$0\leq h(Y_0) \leq 1/k$, and (3)~$h(1/k)\leq 1/k$. (2) and (3) are clear, and (1) simplifies to
\begin{equation}\label{onE}
\mu \leq \frac{1-k Y_0}{k^2 Y_0}.
\end{equation}

We assume (\ref{Vshape}) and (\ref{onE}) in the following. Now, to make $h\in \tilde{\mathfrak{G}}$, we need: (1)~$h(0)>0$ (this is clear) and (2)~$h(Y)<Y$ for all $Y\in [Y_0,1/k)$. We see that (2) is easy since $h(Y_0)=\beta Y_0 < Y_0$ (since $0<\beta<1$) and $h(Y)=\beta Y$ on $(Y_0,1/k]$. Note that we have
\begin{equation*}
\frac{1-k Y_0}{k^2 Y_0} - \frac{\beta(1-k Y_0)}{k} = \frac{(1-k\beta Y_0)(1-k Y_0)}{k^2 Y_0} > 0.
\end{equation*}
So, combining (\ref{Vshape}) and (\ref{onE}), we obtain
\begin{lem}\label{firstLemDS}
If $\frac{\beta(1-k Y_0)}{k}< \mu \leq \frac{1-k Y_0}{k^2 Y_0}$, then $h\in \tilde{\mathfrak{G}}$.
\end{lem}

Next, we show that 
\begin{lem}\label{secondLemDS}
$h^2(Y_0) > Y_0$ if and only if $\frac{(1+\beta)(1-k Y_0)}{k}>\mu$.
\end{lem}
\begin{proof}
We have 
\begin{equation*} 
h^2(Y_0)-Y_0 = h(\beta Y_0)-Y_0 = \frac{Y_0 (1-\beta)(1-Y_0 k)(Y_0 \beta k + Y_0 k - \beta + k\mu -1)}{1-Y_0 k}.
\end{equation*}
The last expression is positive if and only if $Y_0 \beta k + Y_0 k - \beta + k\mu -1>0$ since all the other terms are positive. Now a simple calculation shows that the last inequality is equivalent to $\frac{(1+\beta)(1-k Y_0)}{k}>\mu$. 
\end{proof}

For the rest of the paper, we assume $\frac{(1+\beta)(1-k Y_0)}{k}>\mu$. It is clear from the shape of the graph of $h$ (V-shape with $\beta<1$) that if $h$ has a fixed point on $[0,1/k]$, it must be on $[0,Y_0]$. Solving $Y=h(Y)$ on $[0,Y_0]$, we obtain $Y= \frac{\mu Y_0 k}{Y_0\beta k -Y_0 k-\beta+k\mu +1}$ (the unique fixed point of $h$). We write $\tilde z$ for this fixed point. Now, we show
\begin{lem}\label{thirdLemDS}
$\tilde{\Pi}=\{ \tilde{z} \}$. 
\end{lem}
\begin{proof}
Recall that $\tilde{\Pi}=\{ Y \in [0, Y_0]: h(x)\in [0, Y_0] \textup{ and } h^2(Y)=Y \}$. So, we only need to consider $h$ on $[0, Y_0]$, thus we can set $h(Y)= \left(\beta - \frac{\mu k}{1-k Y_0}\right) Y + \frac{\mu}{1-k Y_0} - \mu$. Solving $h^2(Y)=Y$, we obtain $Y=\tilde{z}$.  
\end{proof}

Finally, we show that 
\begin{lem}\label{forthLemDS}
$h^3(Y_0) > \tilde{z}$ if and only if $\frac{(1-k Y_0)(1+2 \beta^2+\sqrt{4\beta^2+1})}{2\beta k} < \mu$ (Likewise, $h^3(Y_0) \geq \tilde{z}$ if and only if $\frac{(1-k Y_0)(1+2 \beta^2+\sqrt{4\beta^2+1})}{2\beta k} \leq \mu$).
\end{lem}
\begin{proof}
Again, we consider the strict inequality case only. Since we assume that $h^2(Y_0)>Y_0$ (Lemma~\ref{secondLemDS}), we have $h^3(Y_0)=h(h^2(Y_0))=\beta h^2(Y_0)=h(\beta Y_0)$. Now a (slightly messy) calculation gives
\begin{equation}\label{lastEqnDS}
h(\beta Y_0)-\tilde{z} = \frac{Y_0(1-\beta)\left(\beta k^2 \mu^2 +\mu(2Y_0\beta^2 k^2 + Y_0 k^2 -2\beta^2 k - k) + \beta^3 + {Y_0}^2\beta^3 k^2 -2 Y_0 \beta^3 k\right)}{(1-Y_0 k)\left((1-\beta)(1-Y_0 k) + k\mu\right)}.
\end{equation}
Now, we see that the last expression of (\ref{lastEqnDS}) is positive if and only if 
\begin{equation*}
\beta k^2 \mu^2 +\mu(2Y_0\beta^2 k^2 + Y_0 k^2 -2\beta^2 k - k) + \beta^3 + {Y_0}^2\beta^3 k^2 -2 Y_0 \beta^3 k>0.
\end{equation*} 
Solving the last inequality for $\mu$ (and considering that $\mu>0$), we obtain $\frac{(1-k Y_0)(1+2 \beta^2+\sqrt{4\beta^2+1})}{2\beta k} < \mu$.
\end{proof}

Now, thanks to Proposition~\ref{ChaosThm2}, combining Lemmas~\ref{firstLemDS},~\ref{secondLemDS},~\ref{thirdLemDS}, and~\ref{forthLemDS}, we obtain Theorem~\ref{main2}.

%% file: ergodic1.tex
\section{Ergodic properties for macrodynamics}\label{SecErgodic}
Our (numerical/theoretical) argument in the rest of the paper use (a tiny bit of) ergodic theory. Here, we give a quick review of ergodic theory. If the reader is familiar with ergodic theory, skip Subsection~\ref{background}. Our basic references for ergodic theory are classical~\cite{ColletEckmann-dynamics-book},~\cite{Day-dynamics-book}, and~\cite{MeloStrien-dynamics-book}. Note that our strategy (philosophy) in the following stems from~\cite{Lyubich-forty-Modern} and~\cite{Shen:2014} (these are quite readable expository articles on recent developments of unimodal dynamics). We stress that a deep result by Avila et.al.~(Proposition~\ref{Avila}) theoretically supports our argument. Note that our argument in Sections~\ref{SecErgodic} and~\ref{SecErgodicSensitivity} is very similar to (actually it is almost a replicate of) that in our unpublished note~\cite[Secs.~5 and~6]{Uchiyama-ergodic-arX}, but we include it here to make the paper self-contained. (Still the same argument applies here anyway.)

\subsection{Background from ergodic theory}\label{background}
Let $I$ be a compact interval on $\mathbb R$.  Let $\mathcal{B}$ denote the Borel $\sigma$-algebra of $I$. Let $\zeta: \mathcal{B}\rightarrow [0,\infty]$ be a measure on $I$. A measurable map $g: I\rightarrow I$ is called \emph{ergodic with respect to $\zeta$} if whenever $g^{-1}(A)=A$ for $A\in \mathcal{B}$, then either $\zeta(A)=0$ or $\zeta(I\backslash A)=0$. We say that a measure $\zeta$ is \emph{g-invariant} if $\zeta(g^{-1}(A))=\zeta(A)$ for any $A\in \mathcal{B}$. We write $\lambda$ for the Lebesgue measure on $I$. A measure $\zeta$ is called \emph{absolutely continuous with respect to $\lambda$} if whenever $\lambda(A)=0$ for $A\in \mathcal{B}$ then $\zeta(A)=0$. We write an "acim" for a measure that is absolutely continuous with respect to $\lambda$ and $g$-invariant (if $g$ is clear from the context, we just say "invariant"). Remember that if $\zeta$ is an acim then there exists a ($\lambda$-integrable) density function (Radon-Nikodym derivative) $\xi: I \rightarrow \mathbb{R}$ with $d\zeta = \xi d\lambda$, see~\cite[\Rmnum{2}.8]{ColletEckmann-dynamics-book}. It is well-known that if $g$ is ergodic with respect to an acim $\zeta$, then we can give an estimate of $\xi$ (hence an estimate of $\zeta$) using some iterates of $g$, see~\cite[8.5.2 and 8.5.3]{Day-dynamics-book} (we use this argument below in Theorem~\ref{ergThm1-1}).  

Here, we recall (a special case of) the famous Birkhoff's ergodic theorem~\cite[Thm.~8.2]{Day-dynamics-book}, which states that under certain conditions the time average of $g$ is equal to its space average:
\begin{prop}\label{birkhoff}
If $g$ has an acim $\zeta$, $g$ is ergodic with respect to $\zeta$, and $g$ is $\zeta$-integrable, then
\begin{equation}\label{ergodicSumEqn}
\lim_{n\rightarrow \infty}\frac{1}{n}\sum_{k=0}^{n-1}g^{k}(x)=\int_I g\; d\zeta \text{ for $\zeta$-almost all $x$.}
\end{equation}
\end{prop} 
 
In particular, Proposition~\ref{birkhoff} says that (if the conditions are met) the time average of $g$ converges to a constant (for almost all $x$), in other words, we can "predict" the future on average. We call the integral on the right-hand side (or sometimes the sum inside the limit on the left-hand side) of Equation (\ref{ergodicSumEqn}), that is $\int_I g\; d\zeta$ (or sometimes $\frac{1}{n}\sum_{k=0}^{n-1}g^{k}(x)$), the \emph{ergodic sum} of $g$ (which we are referring to would be clear from the context). In the following, we try to apply Proposition~\ref{birkhoff} to our macrodynamics defined by $f$ (or $h$) with the condition $f\in \mathfrak{G}$ (or $h\in \tilde{\mathfrak{G}}$ respectively). To do that, we need to show that $f$ (or $h$) has an absolutely continuous invariant measure $\zeta$, $f$ (or $h$) is ergodic with respect to $\zeta$, and $f$ (or $h$) is $\zeta$-integrable. (The last condition is clear since $f$ (or $h$) is continuous on a compact interval.)  

\begin{rem}
As we said in Introduction, we need to use two quite different mathematical techniques to apply Proposition~\ref{birkhoff} to $f$ and $h$ respectively. For $f$ ("nonlinear model"), we need some very delicate combination of argument from the theory of $S$-unimodal maps as in Subsections~\ref{SubSunimodal}. For $h$ ("piecewise linear model"), a classical result of Lasota and Yorke~\cite{LasotaYorke-expansive-Trans} for iteratively expansive maps gives a big shortcut, see Section~7 for details.  
\end{rem}

\subsection{$S$-unimodal maps}\label{SubSunimodal}
In general, it is pretty difficult to prove the existence of an acim for a measurable transformation $g$ except for some special cases such as when $g$ is "expansive" or "iteratively expansive", see~\cite[Chap.~5]{MeloStrien-dynamics-book},~\cite[Sec.~4]{Shen:2014}, and~\cite{Lyubich-forty-Modern} for an overview of this problem. Recall that $g$ is called \emph{expansive} if $g$ is piecewise $C^2$ and $|g'(x)|>1$ for $\lambda$-almost all $x$. A typical example of an expansive map is a well-studied "tent map"~\cite[8.5.4 and 8.5.5]{Day-dynamics-book}. Further recall that a slightly more general "iteratively expansive" map, that is a piecewise $C^2$ map with $|(g^n(x))'|>1$ for some positive integer $n\geq 1$ for $\lambda$-almost all $x$. In Section~7, we consider a numerical example where our function $h$ is iteratively expansive. (This is actually an easy case.) See~\cite{SatoYano-ergo-AIP} and~\cite{SatoYano-chaos-IJET} for applications of an acim for an iteratively expansive map in economics. 

In the following, we consider the function $f$ (nonlinear model) first. Note that $f$ is not (even iteratively) expansive since it has a critical point $s$. (So this is a hard case.) It turns out that to establish the existence of an acim for $f$, we need some deep analytical results to investigate the counter play between the contraction of $f$ near the critical point $s$ and the expansion of $f$ at $f(s)$ (that is far from $s$). Here, the main tool we use is the theory of $S$-unimodal maps (due to Singer~\cite{Singer-unimodal-SIAM}), whose ergodic properties are well studied, see~\cite[Part \Rmnum{2}]{ColletEckmann-dynamics-book},~\cite[Chap.~5]{MeloStrien-dynamics-book},~\cite{Avila-unimodal-Annals},~\cite{Avila-stochastic-Invent} for example. (Our function $f$ is actually $S$-unimodal as we will show below.) Let $g$ be a measurable transformation defined on a compact interval $I=[a,b]$ of $\mathbb{R}$.
\begin{defn}
A function $g$ is called \emph{$S$-unimodal} if the following conditions are satisfied:
\begin{enumerate}
\item{$g$ is $C^3$.}
\item{$g$ is unimodal with the unique critical point $x=c$ in $(a,b)$ and $g'(x)\neq 0$ except when $x=c$.}
\item{The Schwarzian derivative of $g$, that is, $Sg(x)=\frac{g'''(x)}{g'(x)}-\frac{3}{2}\left(\frac{g''(x)}{g'(x)}\right)^2$ is negative except at $x=c$.}
\end{enumerate}  
Moreover the critical point $x=c$ is called \emph{non-flat} and of order $l$ if there are positive constants $O_1$, $O_2$ with 
\begin{equation*}
O_1|x-c|^{l-1} \leq |g'(x)| \leq O_2|x-c|^{l-1}.
\end{equation*}
\end{defn}

Note that a well-known "logistic map" ($g(x)=r x(1-x)$ for $r\in (0,4]$) is $S$-unimodal. The first key result in this section is
\begin{prop}\label{ergodProp}
If $g$ is $S$-unimodal with a non-flat critical point and without an attracting periodic orbit, then $g$ is ergodic with respect to any absolutely continuous measure $\zeta$. 
\end{prop}
\begin{proof}
The following proof is in~\cite{Uchiyama-ergodic-arX}, but, we reproduce it here since the proof is short. Let $g$ be $S$-unimodal with a non-flat critical point and without an attracting periodic orbit. Suppose that $g^{-1}(A)=A$ for some $A\in\mathcal{B}$. Then we have $\lambda(g^{-1}(A))=\lambda(A)$. From~\cite[Thm.~1.2]{MeloStrien-dynamics-book} we know that $g$ is ergodic with respect to $\lambda$, so we obtain that $\lambda(A)=0$ or $\lambda(I\backslash A)=0$. This yields $\zeta(A)=0$ or $\zeta(I\backslash A)=0$ since $\zeta$ is absolutely continuous.  
\end{proof}

To end this subsection, we prove that
\begin{lem}\label{S-unimodalLem}
The function $f \in \mathfrak{G}$ is $S$-unimodal and the unique critical point $s \left(=\frac{\beta+\delta}{2\delta k}\right)$ of $f$ is non-flat and of order $2$.
\end{lem}
\begin{proof}
First, it is clear that $f$ is $C^3$ and unimodal with the unique critical point $Y=s$ in $(0,1/k)$ and $f'(Y)\neq 0$ except at the critical point. Second, since $f'''(Y)=0$ and $f''(x)\neq 0$ everywhere, we have $Sf(Y)<0$ except at the critical point. Third, since we have $|f'(Y)|=|2\delta k Y -\beta-\delta|$ (note that the turning point of $|f'(Y)|$ is at $Y=s$), if we take $O_1=\delta k$ and $O_2=3\delta k$, we can bound $|f'(Y)|$ from below and above by $O_1|Y-s|$ and $O_2|Y-s|$ as required. 
\end{proof}

\subsection{Our strategy and the existence of an acim}\label{SubOurStrategy}
The second key result in this section is

\begin{prop}\cite[Thm.~\Rmnum{2}.4.1]{ColletEckmann-dynamics-book}\label{criticalOrbit}
If $g$ is $S$-unimodal, then every stable periodic orbit attracts at least one of $a$, $b$, or $c$ (i.e.~the endpoints of $I$ or the critical point of $g$). 
\end{prop}
Proposition~\ref{criticalOrbit} means that all "visible" orbits (in numerical experiments) are orbits containing $a$, $b$, or $c$ only (in the long run). We consider that only these visible orbits are meaningful in economics (or in real life) since it is widely believed that every economic modelling is some sort of an approximation of real economic activities and contains inevitable errors. Here is the third key result for this section: 
\begin{prop}\cite[Cor.~\Rmnum{2}.4.2]{ColletEckmann-dynamics-book}\label{noStableOrbit}
If $g$ is $S$-unimodal, then $g$ has at most one stable periodic orbit, plus possibly a stable fixed point. If the critical point $c$ is not attracted to a stable periodic orbit, then $g$ has no stable periodic orbit.   
\end{prop}

In this paper, we interpret our numerical calculations based on  Propositions~\ref{criticalOrbit} and~\ref{noStableOrbit}. In particular, we look at the orbit starting from the critical point $s$, that is $\{s, f(s), f^2(s),\cdots\}$ (we call this orbit the "critical orbit"). If the critical orbit seems to eventually converge to a periodic orbit, we conclude that we can see the future: the average GDP level in the long run will be the average GDP level in this attracting periodic orbit. Note that in this case, $f$ is neither ergodic nor has an acim since most $f^n(s)$ accumulate around this attracting periodic orbit, but we do not care (since we can still predict the future). Otherwise, we compute (or give an estimate for) the following \emph{Lyapunov exponent} at the critical point $Y=s$ since the existence of a positive Lyapunov exponent at the critical point implies that the critical orbit is repelling and also is a strong indication for the existence of a chaos (hence the existence of an acim, see (CE1) in Proposition~\ref{sufficientACIM} below):

\begin{defn}
$\lim_{n\rightarrow \infty}\frac{1}{n}\sum_{i=1}^{n}\ln{|Dg^{n}(g(c))|}$ is called the \emph{Lyapunov exponent} of $g$ at $c$ (if the limit exists).
\end{defn}

If the Lyapunov exponent (at $c$) is positive, we test one of the following well-known sufficient conditions (within some numerical bound) to confirm the existence of an acim. 

\begin{prop}\label{sufficientACIM}
Suppose that $g$ is $S$-unimodal, $g$ has no attracting periodic orbit, and the critical point $c$ is non-flat. Then
$g$ has a unique acim $\zeta$ and $g$ is ergodic with respect to $\zeta$ if one of the following conditions is satisfied:
\begin{enumerate}
\item{Collet-Eckmann conditions (together with some regularity conditions)~\cite[Sec.~1]{ColletEckmann-Liapunov-ETDS}:
\begin{alignat*}{2}
&\textup{(CE1)}\;\;\liminf_{n\rightarrow\infty}\frac{1}{n}\ln\left| Dg^n(g(c)) \right| > 0, \\
&\textup{(CE2)}\;\;\liminf_{n\rightarrow\infty}\frac{1}{n}\inf_{Y\in g^{-n}(c)}\ln\left| Dg^n(Y)) \right| > 0. 
\end{alignat*} 
} 
\item{Misiurewicz condition~\cite[Thms.~6.2 and~6.3]{Misiurewicz-acim-IHES}: The $\omega$-limit set of $c$ does not contain $c$, that is, 
\begin{equation*}
(MC)\;\; c\notin\bigcap_{n\geq 0}\overline{\{g^i(c): i\geq n\}}. 
\end{equation*}
}
\item{Nowicki-Van Strien summation condition~\cite[Main Thm.]{Strien-SC-Invent}: if $l$ is the order of $c$, then
\begin{alignat*}{2}
&(SC)\;\; \sum_{n=1}^{\infty}|Dg^n(g(c))|^{-1/l} < \infty.
\end{alignat*}
}
\end{enumerate}
\end{prop}
Note that the ergodic property in Proposition~\ref{sufficientACIM} follows from Proposition~\ref{ergodProp}. Now, if one of the conditions in Proposition~\ref{sufficientACIM} is satisfied (within some numerical limitation), we conclude that we can predict the future by Proposition~\ref{birkhoff}. We must admit that our argument in this section is not rigorous (we hope to make it rigorous in the future), but we believe that we have provided enough (numerical/theoretical) evidence to support it. We stress that it is very hard to prove the existence of an acim for any non-expansive function $g$ (even for an $S$-unimodal $g$) by a rigorous analytic argument. There are only a few known examples of such, see a famous $g(x)=4x(1-x)$ example due to Ulam and Neumann~\cite{Ulam-ACIM-BullAMS}, also see~\cite[Sec.~7 Examples]{Misiurewicz-acim-IHES} for more examples. 

In this paper, we test Condition 3 (SC) in Proposition~\ref{sufficientACIM} since the sum in (SC) is easy to compute (estimate) numerically. See~\cite[4.2]{Shen:2014} for a comparison of these three sufficient conditions for the existence of an acim, also see~\cite[Chap.~\Rmnum{5}, Sec.~4]{MeloStrien-dynamics-book} for more on (SC). We found that (CE2) was hard to compute since the set $g^{-n}(c)$ can be very large for a large $n$. Also, the $\omega$-limit set of $c$ was difficult to compute for us although (MC) is theoretically beautiful. (For a numerical computation, it is not clear where to set the numerical bound to estimate the $\omega$-limit set.)

\begin{rem}
Roughly speaking, all three condtions (CE1), (MC), and (SC) are basically testing the same thing: they (more or less) guarantee that the critical orbit does not accumulate around the critical point $c$. (So, on the critical orbit, the expansion far from the critical point wins against the contraction near the critical point.)
\end{rem}

\subsection{$\delta=3.7$ case}\label{SubDelta3.7}
For the rest of the paper, to simplify the exposition and to obtain sharp numerical results, we fix $(\beta, k, \delta)=(0.1,1.1,3.7)$  (in the next Section, we let $\delta$ vary). In this case, the condition for $f\in\mathfrak{G}$, that is, $\max\{2-\beta,\beta\}<\delta<2-\beta+2\sqrt{1-\beta}$ becomes $1.9<\delta < 3.8$ and it is clearly satisfied. Also, from Theorem~\ref{main}, we see that the dynamics exhibit a Li-Yorke chaos (since we have $3.68-\beta<\delta$). In this case, we have $E=[0,1/k]\approx[0,0.91]$. Here is our main result in this section. (We can predict the future even if $f$ is chaotic.)

\begin{thm}\label{ergThm1-1}
There exists an acim $\zeta$ (whose estimate is as in Figure~\ref{fig8}) for $f$. Moreover, we have $\lim_{n\rightarrow \infty}\frac{1}{n}\sum_{k=0}^{n-1}f^{k}(p)\approx 0.6$ for $\zeta$-almost all $Y\in E$.
\end{thm} 

Now, we start looking at the model closely. First, following our strategy as in the last subsection, we consider the critical orbit of $f$. We have that the critical point of $f$ is $s=\frac{\beta+\delta}{2\delta k}=0.46683$. Using the first $100$ iterates of $f^n(s)$, we obtain Figure~\ref{fig5} that shows a chaotic behaviour of the iterates of $f$. 
\begin{figure}[h!]
	\begin{center}
    	\includegraphics[scale=0.55]{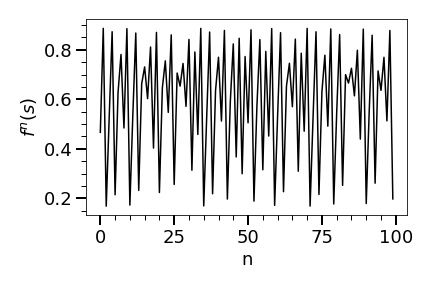}
	\end{center}
    \caption{The first $100$ iterates of $f^n(s)$ look chaotic}\label{fig5}
\end{figure}
In particular, it seems that $f$ has no attracting periodic orbit. To convince the reader that this is really the case, we give an estimate for the Lyapunov exponent (using the first 10000 terms of $f^n(s)$). Figure~\ref{fig6} shows that the first 1000 terms are enough to estimate the Lyapunov exponent (but we used the first 10000 terms to be safe). We obtain
\begin{lem}
$\lim_{n\rightarrow \infty}\frac{1}{n}\ln{|Df^{n}(f(s))|}\approx \frac{1}{10000}\ln{|Df^{10000}(f(s))|}\approx 0.43>0$.
\end{lem}  

\begin{figure}[h!]
	\begin{center}
    	\includegraphics[scale=0.55]{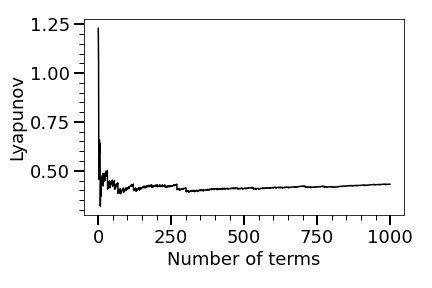}
	\end{center}
    \caption{Convergence of the Lyapunov exponent}\label{fig6}
\end{figure}

Next, we check the Nowicki-Van Strien summation condition (SC). Figure~\ref{fig7} shows that the sum in (SC) stabilises if we use the first 100 terms. To be safe, we use 1000 terms to estimate the infinite sum in (SC). We obtain 

\begin{figure}[h!]
	\begin{center}
    	\includegraphics[scale=0.55]{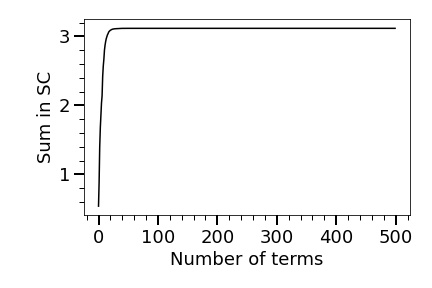}
	\end{center}
    \caption{Convergence of the sum in SC}\label{fig7}
\end{figure}

\begin{lem}\label{SCLem}
$\sum_{n=1}^{\infty}|Df^n(f(s))|^{-1/l}\approx \sum_{n=1}^{1000}|Df^n(f(s))|^{-1/2}\approx 3.11628419346255$.
\end{lem}
\begin{rem}
Since Lemma~\ref{SCLem} is crucial for our argument, we have double-checked (SC) with 100000 terms obtaining  $\sum_{n=1}^{100000}|Df^n(f(s))|^{-1/l}\approx 3.11628419346255$. (This is the same number as in Lemma~\ref{SCLem}!.)
\end{rem}
Now, we conclude that $f$ has a unique acim $\zeta$ and $f$ is ergodic with respect to $\zeta$ by Proposition~\ref{sufficientACIM}. Next, in Figure~\ref{fig8} using $f^n(s)$ with $n$ from $1000$ to $10000$ (after removing the effect of a transient period) we obatin an estimate of the density function (Radon-Nikodym derivative) $\xi: E \rightarrow \mathbb{R}$ with $d\zeta = \xi d\lambda$. Note that although we give an estimate of $\zeta$ in Figure~\ref{fig8}, it is hard to give an explicit formula for $\zeta$ (thus it is hard to express $\zeta$ in a concrete way). 

\begin{figure}[h!]
	\begin{center}
    	\includegraphics[scale=0.55]{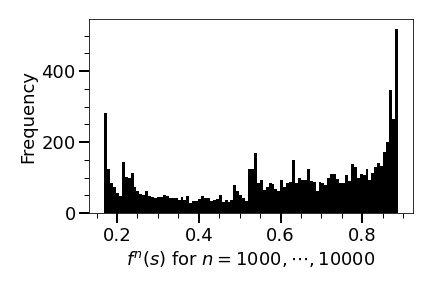}
	\end{center}
    \caption{Estimate of the density function}\label{fig8}
\end{figure}

Next, we directly compute the ergodic sum of $f^n(Y)$ using the initial $Y=s$ and the first $100000$ terms of $f^n(s)$. We get
\begin{lem}
 $\lim_{n\rightarrow \infty}\frac{1}{n}\sum_{k=0}^{n-1}f^{k}(s)\approx \frac{1}{100000}\sum_{k=0}^{99999}f^{k}(s)\approx 0.6$. 
 \end{lem}
Note that Figure~\ref{fig9} shows that the ergodic sum converges if we use more than $2000$ terms to estimate it (we use $100000$ terms to be safe).  
\begin{figure}[h!]
	\begin{center}
    	\includegraphics[scale=0.55]{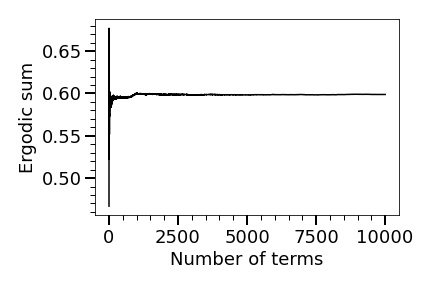}
	\end{center}
    \caption{Convergence of the ergodic sum of $f^n(s)$}\label{fig9}
\end{figure}
To end this section, we compute the ergodic sums using various initial values. (These numbers should be the same for almost all initial values by Proposition~\ref{birkhoff}.) Figure~\ref{fig10} shows the result where the initial $Y$ is taken from $0.01, 0.02, \cdots, 0.91$. By Figure~\ref{fig10} (where each ergodic sum is estimated using the first $5000$ terms of $f^n(s)$), we conclude that Theorem~\ref{ergThm1-1} holds and the ergodic sums of $f^n(Y)$ converge to somewhere around $0.6$ for $\zeta$-almost all $Y\in E$. 

\begin{figure}[h!]
	\begin{center}
    	\includegraphics[scale=0.55]{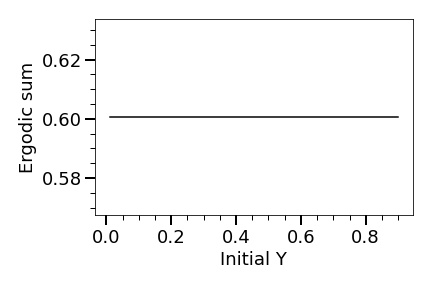}
	\end{center}
    \caption{Ergodic sums of $f^n(Y)$ using various initial $Y$}\label{fig10}
\end{figure}





%% file: ergodic2.tex
\section{Ergodic properties: a sensitivity analysis}\label{SecErgodicSensitivity}
\subsection{Chaos is not too bad}\label{SubChaosNotBad}
In this section, we still keep $(\beta, k)=(0.1,1.1)$, but we let $\delta$ vary. Here, we conduct a sensitivity analysis, that is, how the ergodic sums vary when $\delta$ changes for $1.9<\delta\leq 3.8$. We note that the (unique) critical point of $f$, that is $s=\frac{\beta+\delta}{2\delta k}$, is dependent of $\delta$. Also note that from Subsection~2.3, we know that there exists an odd period (hence a Li-Yorke chaos) for $3.58<\delta\leq 3.8$.

Now, using the same strategy as in the last section, we investigate how the critical orbit $\{s, f(s), f^2(s),\cdots\}$ behaves for each $\delta$. We obtain the bifurcation diagram (a summary of the critical orbits) of $f$ in Figure~\ref{fig11}. We add Figure~\ref{CloseUp}, that is a close-up view of Figure~\ref{fig11} around the "chaotic region". We (roughly) see that: (1)~for $\delta<2.9$, $f^n(s)$ converges to the unique attracting fixed point, (2)~for $2.9<\lambda<3.35$, $f^n(s)$ converges to a period-$2$ orbit, (3)~for $3.35<\lambda<3.5$ (roughly), period doubling bifurcations occur and $f^n(s)$ converges to a period-$4$ (8, 16, and so on) orbit, (4)~for $3.5<\delta<3.58$, we see two (upper and lower) chaotic regions, (5)~for $3.58<\delta$, we see one large chaotic region except a few "windows" (thin white strips), (6)~for $\delta\approx 3.71$, we (finally) see a period three cycle.

\begin{rem}
Our argument in the last paragraph (in particular (6)) shows that if we use "Period three implies chaos" only, then we do not see the (large) chaotic region for $3.58<\delta<3.71$. 
\end{rem}

\begin{figure}[h!]
	\begin{center}
   	\includegraphics[scale=0.7]{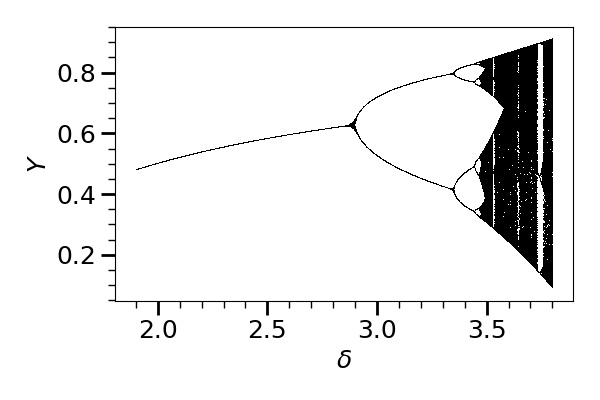}
	\end{center}
    \caption{The bifurcation diagram of $f$}\label{fig11}
\end{figure}

\begin{figure}[h!]
	\begin{center}
   	\includegraphics[scale=0.7]{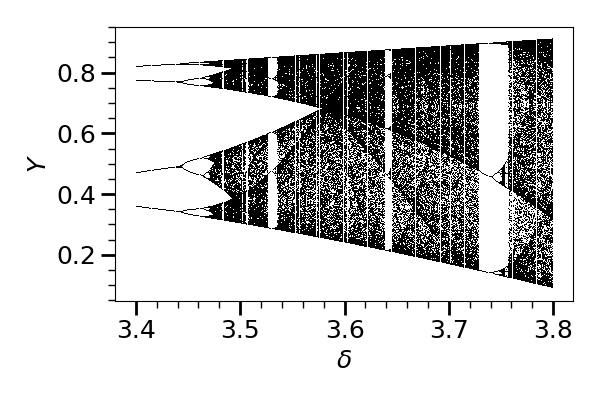}
	\end{center}
    \caption{The bifurcation diagram of $f$ in the chaotic region}\label{CloseUp}
\end{figure}

Next, we compute Lyapunov exponents of $f$ at $s$ for various $\delta$ and obtain Figure~\ref{fig12}. We see that roughly speaking the Lyapunov exponent is negative for $1.9<\delta<3.5$ (thin dots), and positive for $3.5<\lambda\leq 3.8$ (thick dots).

\begin{figure}[h!]
	\begin{center}
    	\includegraphics[scale=0.55]{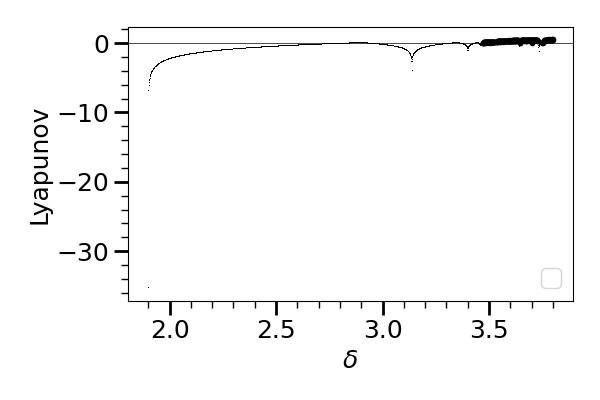}
	\end{center}
    \caption{Lyapunov exponents of $f$}\label{fig12}
\end{figure}

Now, we test (SC) for $3.45<\delta\leq 3.8$. (we do not care $\delta<3.45$ since for these $\delta$, it is clear (from the bifurcation diagram) that $f$ has an attracting cycle and it is easy to predict the future). We obtain Figure~\ref{fig13}.
\begin{figure}[h!]
	\begin{center}
    	\includegraphics[scale=0.55]{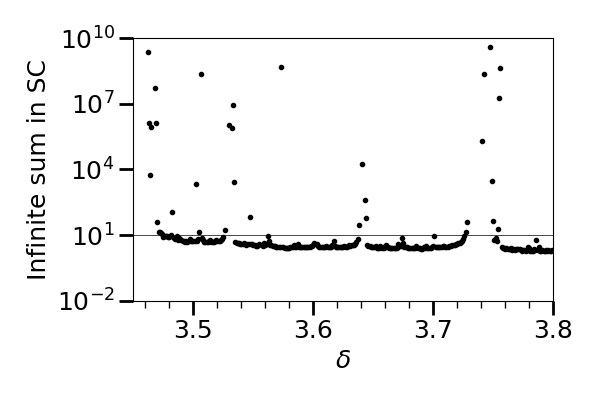}
	\end{center}
    \caption{Infinite sums in SC}\label{fig13}
\end{figure}
Our argument is based on a numerical computation using $1000$ terms to estimate the infinite sum in (SC), so we need to decide when we conclude that the infinite sum is finite. We draw the (ad hoc) line at $10$ (that is the horizontal line in Figure~\ref{fig13}). Basically, we only avoid $\delta$ where (the estimate of) the infinite sum grows exponentially. Our numerical calculation shows that our estimates for the infinite sums are finite for $3.48<\delta$ (roughly) except $\delta$ values corresponding to the few windows in Figure~\ref{CloseUp}. We hope to rigorously prove that the infinite sums here are finite in the future work. Summarising all the argument in this section, we obtain our main result in this section.
\begin{thm}\label{goodLThm}
For $3.48<\delta\leq 3.8$ except $\delta$ values corresponding to the few windows in Figure~\ref{CloseUp} (and possibly except some $\delta$ values whose total Lebesgue measure is $0$, see Proposition~\ref{Avila} below), there exists a unique acim for $f$. Moreover for these $\delta$ values, the ergodic sums of $f$ are as in Figure~\ref{fig14}.  
\end{thm}

For $\delta$ values as in Theorem~\ref{goodLThm} (satisfying the SC), we obtain a pretty smooth relation between $\delta$ and the ergodic sums of $f$ (except a few bumps and a big drop around $\delta=3.75$ that corresponds to a large window in the bifurcation diagram) as in Figure~\ref{fig14} (using $2000$ terms to estimate the ergodic sums). Extending Theorem~\ref{goodLThm} (and Figure~\ref{fig14}) using the naive estimates of the ergodic sum (that is $\sum_{k=0}^{1999}f^k(s)$) for $\delta$ that does not satisfy (SC), we obtain 

\begin{figure}[h!]
	\begin{center}
    	\includegraphics[scale=0.55]{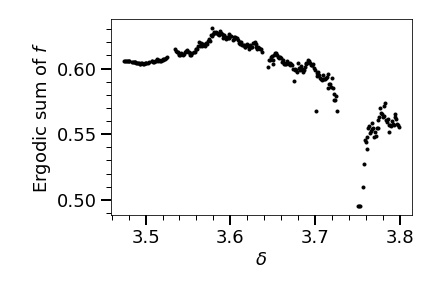}
	\end{center}
    \caption{Ergodic sums of $f$}\label{fig14}
\end{figure}

\begin{thm}\label{finalThm}
For $1.9<\delta<3.8$, the ergodic sums of $f$ are as in Figure~\ref{fig15} (possibly except some $\delta$ values whose total Lebesgue measure is $0$). 
\end{thm}

\begin{figure}[h!]
	\begin{center}
    	\includegraphics[scale=0.55]{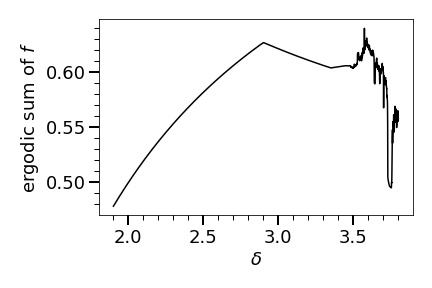}
	\end{center}
    \caption{Average GDP levels under $f$}\label{fig15}
\end{figure}


\subsection{Comments and interpretations on Theorem~\ref{finalThm}}\label{SubCommentsInterpretations}

Here, we add a few comments and interpretations on Theorem~\ref{finalThm}: (1)~To be honest, we find it hard to give any economic interpretation for Fig.~\ref{fig15}, so we just give some mathematical interpretation only. (2)~Mathematically speaking, the changes of the ergodic sums of $f$ for $\delta<3.5$ (before the chaotic region) is clear from the bifurcation diagram (Fig.~\ref{fig11}). (3)~for $\delta>3.5$, we see a gradual decrease of the ergodic sums. To see why, we need to look closer, namely, we need to investigate the density $\xi$ for each $\delta$. We compute the estimates of distributions of $f^n(s)$ for $\delta=3.59, 3.7, 3.79$ using $10000$ terms in Figures~\ref{fig16},~\ref{fig8},~\ref{fig17} respectively.

Our computation shows that: (1)~The "shapes" of the distributions of $f^n(s)$ for various $\delta$ look similar: we see almost uniform distributions with some concentrations at both extremes, (2)~$f$ takes more extreme values as $\delta$ becomes large, however, the extension of the lower bound is much greater than that of the upper bound as shown in Figures~\ref{fig16},~\ref{fig8}, and~\ref{fig17}, (3)~The distribution gets smoother as $\delta$ becomes large. We conclude that (1) and (2) above makes the ergodic sums decrease gradually as $\delta>3.5$ becomes large. We do not know why the density curve gets smoother as $\delta$ increases. 

\begin{figure}[h!]
	\begin{center}
    	\includegraphics[scale=0.55]{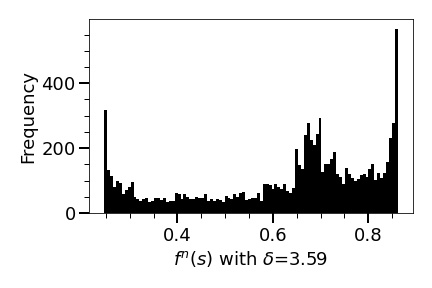}
	\end{center}
    \caption{Density of $f^n(s)$ with $\delta=3.59$}\label{fig16}
\end{figure}

\begin{figure}[h!]
	\begin{center}
    	\includegraphics[scale=0.55]{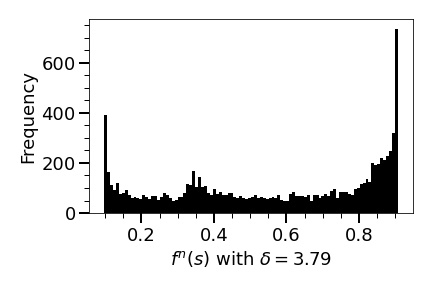}
	\end{center}
    \caption{Density of $f^n(s)$ with $\delta=3.79$}\label{fig17}
\end{figure}


Theorem~\ref{finalThm} (and the whole results in this paper) says that a naive estimate of the ergodic sums of $f$ (estimate of the future) using a reasonably large number of terms (say $2000$ terms) is not too bad. To end the paper, we quote a deep result of Avila (2014 fields medalist) and others~\cite[Sec.~3.1, Theorem B]{Avila-stochastic-Invent} that supports Theorem~\ref{finalThm}. 
\begin{prop}\label{Avila}
In any non-trivial real analytic family of quasiquadratic maps (that contains $S$-unimodal maps), (Lebesgue) almost any map is either regular (i.e., it has an attracting cycle) or stochastic (i.e., it has an acim). 
\end{prop}
\begin{rem} A technical note: "non-triviality" is guaranteed for our set of maps $f$ (parametrised by $\delta$) since there exist two maps in this set that are not topologically conjugate. For example, take $f$ with $\delta=2.5$ (non-chaotic) and $f$ with $\delta=3.7$ (chaotic). See~\cite[Sec.~2.8, Sec.~2.13, Sec.~3.1]{Avila-stochastic-Invent} for the precise definitions of "non-trivial" and "quasiquadratic" (those are a bit too technical to state here). Also, see~\cite{Avila-unimodal-Annals}, and~\cite{Lyubich-forty-Modern} for more on this.
\end{rem}
\begin{rem}
We need "Lebesgue almost" (or "except a set of measure zero") in Theorems~\ref{goodLThm} and~\ref{finalThm}, and Proposition~\ref{Avila} since the following (anomalous) examples are known, see~\cite{Hofbauer-asymptotic-MathPhys} and~\cite{Johnson-singular-MathPhys}: for a quadratic map $T_\gamma(x)=\gamma x(1-x)$ (parametrised by $\gamma$), there exists $\gamma$ such that $T_\gamma$ does not have an attracting periodic orbit and shows a chaotic behaviour, but does not have an acim. We expect that for our $f$, we obtain examples of the same properties (although we have not checked yet). The point is that we do not care such anomalous cases since the $\gamma$ values corresponding to such examples are of Lebesgue measure zero and our approach in this (and the last) section is probabilistic.  
\end{rem}

%% file: ergodicIterative.tex
\section{Ergodic properties: an iteratively expansive case}
To finish the paper, we consider our second model ("piecewise linear model") in this section. The line of the argument below is similar to that in the last two section (so we omit some details) but the analysis here is much easier since we can apply Proposition~\ref{LasotaIterative} (a classical result of Lasota and Yorke~\cite[Thm.~3]{LasotaYorke-expansive-Trans} for the existence of an acim for an iteratively expansive map).  

First, we have that $h$ has two slopes, namely, $\beta-\frac{\mu k}{1-k Y_0}$ for $0\leq Y \leq Y_0$ and $\beta$ for $Y_0 < Y \leq \frac{1}{k}$. We write $\beta_1:=\beta-\frac{\mu k}{1-k Y_0}$ to ease the notation. Here, we assume $h \in \tilde{G}$. Then, by Theorem~\ref{main2}, $h$ has an odd period cycle (hence chaotic) if and only if $\frac{(1+\beta)}{1-k Y_0}{k}<\mu$. An easy computation shows that the last inequality is equivalent to $\beta_1<-1$. By the argument in the last two sections, if the critical orbit is not chaotic we can (easily) predict the future GDP level. Therefore, we focus on the (chaotic) case where $\beta_1<-1$. Since we know that $0<\beta<1$, to make $h$ iteratively expansive, it is enough to show that: (1)~$f^n(Y)$ does not stay in $[Y_0, 1/k]$ "for a long time" for any $Y\in [0,1/k]$, (2)~$|\beta_1|$ is large compared to $\beta$. (Expansions win against contractions.) Now, we show
\begin{lem}\label{cannotStayLong}
Let $Y\in [Y_0,\frac{1}{k}]$. If $\beta^n < Y_0 k$ for some $n\in \mathbb{N}$,  then $h^{i}(Y)\in [0, Y_0]$ for some $i\leq n$.
\end{lem}
\begin{proof}
Let $Y_0 \leq Y\leq 1/k$ and $\beta^n < Y_0 k$ for some $n\in \mathbb{N}$. Now, suppose that $h^i(Y)\in [Y_0, 1/k]$ for all $i\leq n$. 
Then, we have $h^n(Y)=\beta^n Y\in [Y_0, 1/k]$. Since $Y\leq 1/k$, this forces $Y_0\leq \frac{\beta^n}{k}$. This is a contradiction.
\end{proof}
Let $i^*$ be the minimum $i$ with $h^{i}(1/k)\in [0, Y_0]$ in Lemma~\ref{cannotStayLong}. Then we obtain
\begin{prop}\label{iterExpProp}
The function $h$ is iteratively expansive if $|\beta_1|\beta^{i^*-1}>1$. 
\end{prop}
\begin{proof}
Recall that a function $g$ is iteratively expansive if it is piecewise $C^2$ and $|(g^m(x))'|>1$ for some positive integer $m\geq 1$ for $\lambda$-almost all $x$. First, it is clear that $h$ is piecewise $C^2$. Second, since $\beta_1<-1$ and $0<\beta<1$, Lemma~\ref{cannotStayLong} gives the desired result. (Starting from any $Y\in [Y_0, 1/k]$, $Y$ must go back to $[0, Y_0]$ in at most $i^*$ steps.)
\end{proof}

Now, recall~\cite[Thm.~3]{LasotaYorke-expansive-Trans} and~\cite[Cor.~8.4]{Day-dynamics-book}:
\begin{prop}\label{LasotaIterative}
Let $g$ be an iteratively expansive map from a compact interval to itself. Then $g$ has an acim $\zeta$ such that $g$ is ergodic with respect to $\zeta$. 
\end{prop}

In the following, we consider the numerical example in Section~\ref{OPCPiecewise} (Remark~\ref{iterativeExample}) to illustrate our method. So, we fix $(\beta, k, Y_0)=(0.6,1.1.0.2)$. We have seen that: (1) If $0.42<\mu \leq 3.22$, then $h\in \tilde{G}$, (2) There exists an odd period cycle if and only if $1.93<\mu$. Now, using the critical orbit, we obtain the bifurcation diagram (Figure~\ref{bifurcationDS}).  
\begin{figure}[h!]
	\begin{center}
   	\includegraphics[scale=0.7]{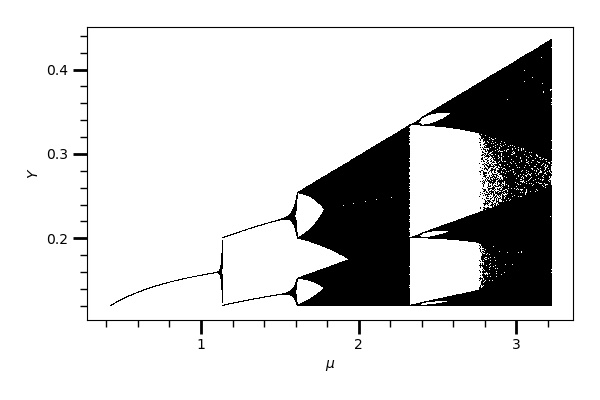}
	\end{center}
    \caption{The bifurcation diagram of $h$}\label{bifurcationDS}
\end{figure}
From the diagram, we see a similar pattern of transitions from attracting cycles of various periods to chaos as $\mu$ becomes large as in the last section. In particular, we see a large  chaotic region for $1.93 < \mu$ (as we expect from Theorem~\ref{main2}). 

Now, we apply Proposition~\ref{iterExpProp}. Note that $h(1/k)=\beta/k \approx 0.545>0.2$, $h^2(1/k)=\beta^2 /k = 0.327>0.2$, $h^3(1/k)= \beta^3/k\approx 0.196 <0.2$. So $i^*=3$, and a simple calculation shows that $|\beta_1|\beta^{2}>1$ is equivalent to 
$\mu > 2.395$ (roughly). Thus, by Proposition~\ref{iterExpProp}, if $\mu > 2.395$, $h$ is iteratively expansive. Then, using Proposition~\ref{LasotaIterative}, $h$ has an acim $\zeta$ and $h$ is ergodic with respect to $\zeta$ for $\mu>2.395$. We conclude that we can use Birkhoff's ergodic theorem to predict the future GDPs for $\mu>2.395$. Note that from the bifurcation diagram, we see that a period three cycle appears at $\mu\approx2.3$. 

Summarising the argument so far (and using the bifurcation diagram), we (roughly) obtain: (1) For $0.42<\mu<1.6$ and for $2.395<\mu\leq 3.22$, we can predict the future GDPs. (2) For $1.6<\mu< 2.395$, we see a chaotic behaviour. In the following, we show that for $1.6<\mu<2.395$, $h$ is iteratively expansive using some numerical calculations. It is known that $h^n$ tends to be (more) expansive as $n$ becomes large if $h$ is iteratively expansive. Here, we take $n=5000$ and obtain Figures~\ref{expansive16},~\ref{expansive161}, and~\ref{expansive18} for $\mu=1.6$, $1.61$, and $1.8$ respectively. We also obtain Figure~\ref{expansive161zoom}, that is a close-up view of Figure~\ref{expansive161}.

\begin{figure}[h!]
	\begin{center}
   	\includegraphics[scale=0.55]{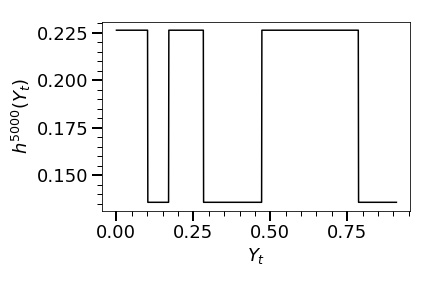}
	\end{center}
    \caption{The graph of $h^{5000}(Y)$ for $\mu=1.6$}\label{expansive16}
\end{figure}

\begin{figure}[h!]
	\begin{center}
   	\includegraphics[scale=0.55]{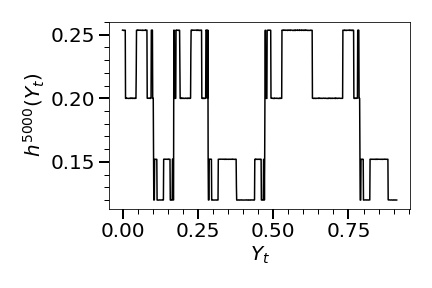}
	\end{center}
    \caption{The graph of $h^{5000}(Y)$ for $\mu=1.61$}\label{expansive161}
\end{figure}

\begin{figure}[h!]
	\begin{center}
   	\includegraphics[scale=0.55]{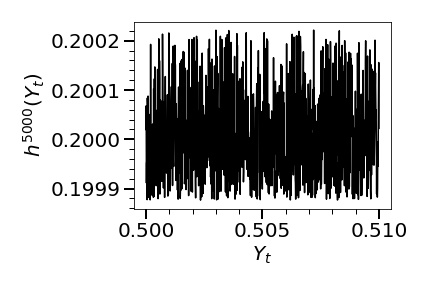}
	\end{center}
    \caption{The (close-up) graph of $h^{5000}(Y)$ for $\mu=1.61$}\label{expansive161zoom}
\end{figure}

\begin{figure}[h!]
	\begin{center}
   	\includegraphics[scale=0.55]{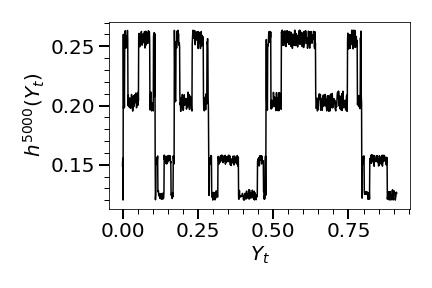}
	\end{center}
    \caption{The graph of $h^{5000}(Y)$ for $\mu=1.8$}\label{expansive18}
\end{figure}

From these figures, we see that: (1) For $\mu=1.6$, $h^{5000}$ is not expansive (its graph has a flat part, see Figure~\ref{expansive16}), (2) If we increase $\mu$ slightly, say, take $\mu=1.61$, then $h^{5000}$ looks non-expansive in Figure~\ref{expansive161}. However, if we look closer, we find that the graph has wiggles everywhere (see Figure~\ref{expansive161zoom}). So, we conclude that $h^{5000}$ is expansive for $\mu=1.61$, (3) If we take $\mu$ larger, the graph of $h^{5000}$ becomes more wiggly, hence expansive as seen in Figure~\ref{expansive18}. Overall, we conclude that for $1.6<\mu<2.395$, $h$ is iteratively expansive. Thus, by the same argument, we can guess the future GDP levels for $1.6<\mu<2.395$.

Now, using the first $2000$ terms of $f^n(Y_0)$ to estimate the ergodic sum for each $\mu$, we obtain
\begin{thm}\label{iterativeThm}
For $0.42<\mu\leq 3.22$, the ergodic sums of $h$ are as in Figure~\ref{AverageDS}. (Possibly except some $\mu$ values whose total Lebesgue measure is $0$.)
\end{thm} 

\begin{figure}[h!]
	\begin{center}
   	\includegraphics[scale=0.55]{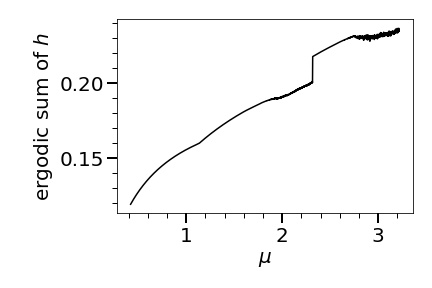}
	\end{center}
    \caption{Average GDP levels under $h$}\label{AverageDS}
\end{figure}

%% file: acknowledgements.tex
\section*{Acknowledgements}
This research was supported by a JSPS grant-in-aid for early-career scientists (22K13904) and an Alexander von Humboldt Japan-Germany joint research fellowship. 